\documentclass[journal]{IEEEtran}

\ifCLASSINFOpdf
\else
   \usepackage[dvips]{graphicx}
\fi
\usepackage{url}

\PassOptionsToPackage{table}{xcolor}
\usepackage[utf8]{inputenc} 
\usepackage[T1]{fontenc}    
\usepackage{hyperref}       
\usepackage{url}            
\usepackage{booktabs}       
\usepackage{amsfonts}       
\usepackage{nicefrac}       
\usepackage{microtype}      
\usepackage[matit]{echodefs}
\usepackage{thmtools, thm-restate}
\usepackage{todonotes}
\usepackage{dirtytalk}
\usepackage{graphicx}
\usepackage{cleveref}
\usepackage{tcolorbox}
\usepackage[bottom]{footmisc}

\usepackage[square,numbers]{natbib}
\definecolor{Gray}{gray}{0.9}
\newcommand{\etal}{\textit{et al. }}

\newcommand{\trace}{\ensuremath{\mathop{\mathrm{Tr}}}}

\def\bL{{\mathbb L}}

\def\bI{{\mathbb I}}

\def\P{{\mathcal P}}
\def\Q{{\mathcal Q}}
\def\bL{{\mathbb L}}
\def\R{{\mathbb R}}
\def\bI{{\mathbb I}}

\def\acosh{{\mathrm{acosh}}}

\begin{document}

\title{On Procrustes Analysis in Hyperbolic Space}

\author{Puoya Tabaghi, Ivan Dokmani\'c, \IEEEmembership{Member, IEEE}
\thanks{Puoya Tabaghi is with the Coordinated Science Lab at University of Illinois at Urbana-Chamapaig (emails: tabaghi2@illinois.edu). Ivan Dokmani\'c is with the Department of Mathematics and Computer Science at University of Basel (email:ivan.dokmanic@unibas.ch)}}


\maketitle

\begin{abstract}
Congruent Procrustes analysis aims to find the best matching between two point sets through rotation, reflection and translation. We formulate the Procrustes problem for hyperbolic spaces, review the canonical definition of the center of point sets, and give a closed form solution for the optimal isometry for noise-free measurements. We also analyze the performance of the proposed method under measurement noise. 
\end{abstract}


\IEEEpeerreviewmaketitle
\section{Introduction}

\IEEEPARstart{I}{n} Greek mythology, Procrustes was a robber who lived in Attica and deformed his victims to match the size of his bed. In 1962, Hurley and Catell used the story of Procrustes to describe a point set matching problem in Euclidean spaces~\cite{hurley1962procrustes}, stated below. 
\begin{problem}
Let $\set{z_n}_{n=1}^{N}$ and $\set{z^{\prime}_n}_{n=1}^{N}$ be two point sets in $\R^d$. The Procrustes problem asks to find a map $\widehat{T}$ that minimizes the sum of the mismatch norms, i.e.,
\[
\widehat{T} = \argmin_{T \in \mathcal{T}} \sum_{n=1}^{N} \norm{z_n - T(z^{'}_n)}_{2}^2
\]
where $\mathcal{T}$ is the set of rotation, reflection, translation, and uniform scaling maps and their compositions~\cite{gower1975generalized}.
\end{problem}

In computer vision, Procrustes analysis is of relevance in \emph{point cloud registration} problems. The task of rigid registration is to find an isometry between two (or more) sets of points sampled from a $2$ or $3$ dimensional object. Point registration has applications in object recognition~\cite{mitra2004registration}, medical imaging~\cite{fitzpatrick1998predicting} and localization of mobile robotics ~\cite{pomerleau2015review}. In signal processing, Procrustes analysis often involved aligning shapes or point sets by a \emph{distance preserving} bijection. Procrustes problems also naturally arise in distance geometry problems (DGPs) where one wants to find the location of a point set that best represents a given set of incomplete point distances, i.e.,
\[
z_1, \ldots, z_N \in \R^d: \norm{z_n - z_m} = d_{mn}, \ \forall (m,n) \in \mathcal{M}
\]
where $\mathcal{M} \subset \{{1,\ldots,N\}}^2$ and $\set{ d_{m,n}: (m,n) \in \mathcal{M}}$ is the set of measured distances~\cite{liberti2014euclidean}. If a distance geometry problem has a solution, it is an orbit of the form
\[
O_\mathcal{Z} = \set{ \set{T(z_n)}_{n=1}^{N} \text{ s.t. } T:\R^d \rightarrow \R^d \mbox{ is an isometry} },
\]
where $\mathcal{Z} = \set{z_n}_{n=1}^{N}$ is a particular solution. In order to uniquely identify the correct solution from all the possible elements in the orbit $O_\mathcal{Z}$, we may be given the exact position of a subset of points, called \emph{anchors}. We use Procrustes analysis to pick the correct solution by finding the best match between the anchors with their corresponding points in the orbit. This technique is commonly used in localization problems~\cite{dokmanic2015euclidean,tabaghi2019kinetic}.

Procrustes analysis can be performed in any metric space. In particular, hyperbolic Procrustes analysis is of great relevance due to the recent surge of interest in hyperbolic embeddings and machine learning~\cite{tabaghi2020hyperbolic,de2018representation}. Furthermore, hyperbolic embeddings are closely connected to the study of hierarchical or tree-like data structures and hyperbolic Procrustes problem solutions may be used to align hierarchical data, e.g., ontologies~\cite{shvaiko2011ontology,euzenat2007ontology}. The goal of ontological studies is to find a (distance preserving) map between a fixed number of entities in two tree-like structures that are best aligned to each other (see \Cref{fig:tree_matching} for an illustration). For example, in \emph{ontology matching} one aims to find correspondences between semantically related entities in heterogeneous ontologies with the goal of ontology merging, query response, or data translation~\cite{shvaiko2011ontology}. 
\begin{figure}[b!] 
	\center
  \includegraphics[width=1 \linewidth]{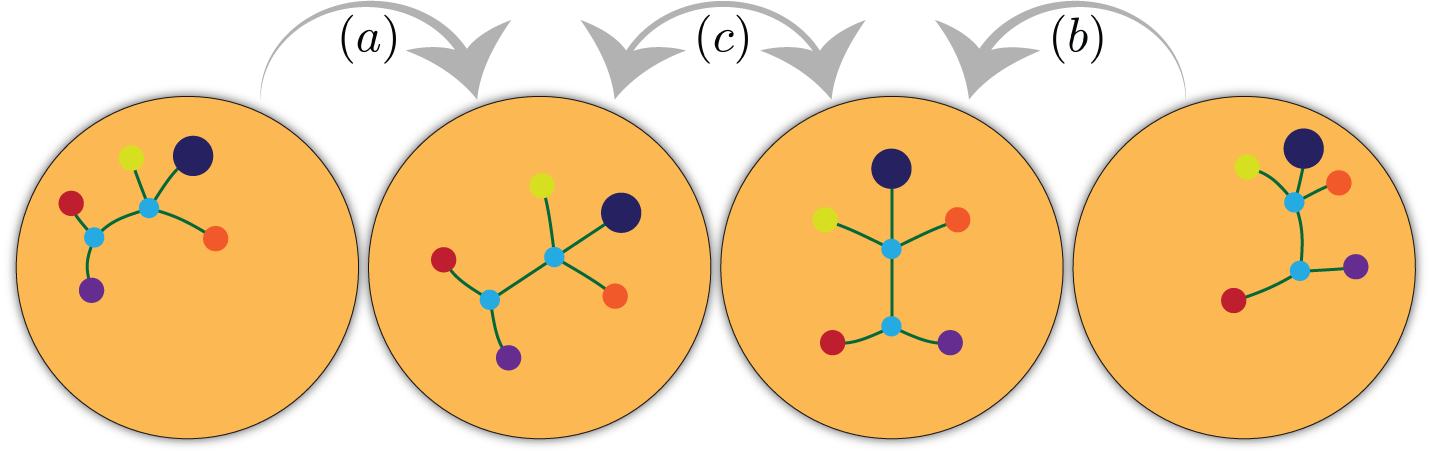}
  \caption{Tree alignment in the Poincar\'e disk \cite{cannon1997hyperbolic}. Hyperbolic Procrustes analysis aims to align two trees, depicted on the far left and far right figures. In steps $(a)$ and $(b)$ we center vertices in both trees, while in step $(c)$ we estimate the unknown rotation map.}
  \label{fig:tree_matching}
\end{figure}

In unsupervised matching problems, the first step in Procrustes-type analyses is to find the correspondence between two point clouds by using the iterative closest point algorithm~\cite{rusinkiewicz2001efficient}. Recently, Alvarez-Melis \etal~\cite{alvarez2020unsupervised} cast the unsupervised hierarchy matching problem in hyperbolic space. Their proposed method jointly learns the \say{soft} correspondence and the alignment map characterized by a hyperbolic neural network. 

In our work, we start with parametric isometries in the 'Loid model of hyperbolic spaces.
It is known that one can decompose any isometry into \emph{elementary} isometries, e.g., hyperbolic translations and hyperbolic rotations (and reflections). In our setting, we aim to find a joint estimate for hyperbolic translation and rotation maps that best align two point sets. 

To accomplish this task, we review the definition of the center of mass, or \emph{centroid}, for a set of points in hyperbolic space. This enables us to subsequently \say{center} each set, and decouple the joint estimation problem into two steps: $(1)$ translate the center of  mass of each point set to the coordinate origin (of the Poincar\'e model), and $(2)$ estimate the unknown rotation factor. While hyperbolic centering has been studied in the literature~\cite{mardia2009directional}, our Procrustes analysis framework is different from prior work in so far that it is similar to its Euclidean counterpart, and provides an optimal estimate for the unknown rotation factor, based on the weighted mean of pairwise inner products. Moreover, we prove a proof that our proposed method ensures the theoretically optimal isometry if the point sets match perfectly. We conclude the paper by giving numerical performance bounds for the task of matching noisy point sets.

\begin{tcolorbox}
{\bf Summary:} Let $\set{x_n}_{n \in [N]}$, and $\set{x^{\prime}_n}_{n \in [N]}$ be two sets of points in a hyperbolic space, related through an isometric map, i.e., $x^{\prime}_n = T(x_n), \forall n \in [N]$. Then,
\[
T = T_{m_{x^\prime}} \circ T_{U} \circ T_{-m_x}
\]
where $m_x, m_y \in \R^d$ are the point sets' centroids, $T_{b}$ is the translation map by vector $b \in \R^d$, and $T_{U}$ is a rotation map by a unitary matrix $U \in \mathbb{O}(d)$; see \Cref{sec:Procrustes_Analysis}. For noisy points, this isometry is suboptimal and can be fine-tuned via a gradient-based algorithm.
\end{tcolorbox}
\emph{Notation.} For $N \in \mathbb{N}$, we let $[N] = \set{1, \ldots, N}$. Depending on the context, $x_1$ can either be the first element of $x \in \R^{d}$, or an indexed vector. We denote the set of orthogonal matrices as $\mathbb{O}(d) = \set{R \in \R^{d \times d}: R^\T R = I}$. For a function $f$ and its inputs $x_1, \ldots, x_N$, we write $\overline{f(x_n)} = \frac{1}{N} \sum_{n \in [N]}f(x_n)$. For a vector $b \in \R^d$, we denote its $\ell_2$ norm as $\norm{b}_2$.
\vspace{-10pt}
\section{'Loid Model of Hyperbolic Space}
Let $x, x^{\prime} \in \R^{d+1}$ with $d \geq 1$. The Lorentzian inner product between $x$ and $x^{\prime}$ is defined as 
\begin{equation}\label{eq:Lorentzian_inner_product}
[x, x^{\prime}] = x^\T H x^{\prime}: H = \begin{pmatrix}
-1 & 0^\T\\
0 & I_d
\end{pmatrix},
\end{equation}
where  $I_d \in \R^{d \times d}$ is the identity matrix. This is an indefinite inner product on $\R^{d+1}$. The vector space $\R^{d+1}$ equipped with the Lorentzian inner product is called a Lorentzian $(d+1)$-space. In a Lorentzian space, we can define notions similar to adjoint and unitary matrices in Euclidean spaces. The $H$-adjoint of the matrix $R$, denoted by $R^{[*]}$, is defined via
\[
[Rx,x^{\prime}] = [x,R^{[*]}x^{\prime}], ~ \forall x, x^{\prime} \in \R^{d+1},
\]
or simply as $R^{[*]} = H^{-1} R^\T H$. An invertible matrix $R$ is called H-unitary if $R^{[*]} = R^{-1}$~\cite{gohberg1983matrices}. 

The 'Loid model of $d$-dimensional hyperbolic space is a Riemannian manifold $\mathcal{L}^d= (\bL^{d} , (g_x)_x )$, where
\[
\bL^{d} = \set{x \in \R^{d+1}: [x,x] = - 1, x_1 > 0}
\]
and the Riemannian metric $g_x: T_x \bL^d \times T_x \bL^d \rightarrow \R$ defined as $g_x(u,v) = [u,v]$. The distance function in the 'Loid model are characterized by Lorentzian inner products as 
\[
d(x, x^{\prime}) = \mathrm{acosh}(-[x, x^{\prime}]), \forall x, x^{\prime} \in \bL^d.
\]
\subsection{Isometries} 
A map $T: \bL^d \rightarrow \bL^d$ is an isometry if it is bijective and preserves distances, i.e. if
\[
d(x,x^{\prime}) = d\big( T(x),T(x^{\prime}) \big),  ~ \forall x,x^{\prime} \in \bL^d.
\]
We can represent any hyperbolic isometry as a composition of two \emph{elementary} maps that are parameterized by a $d$-dimensional vector and a $d \times d$ unitary matrix, as described below. 
\begin{fact} \cite{ratcliffe2006foundations}
\label{fact:1}
The function $T: \bL^d \rightarrow \bL^d$ is an isometry if and only if it can be written as $T(x) = R_{U}R_{b}x$, where 
\[
R_{U} = \left[\begin{array}{cc}
1 & 0^\T\\ 
0  & U 
\end{array}\right], \ R_b = \left[\begin{array}{cc}
\sqrt{1 + \norm{b}_2^2} & b^\T\\ 
b  & (I + b b^\T)^{\frac{1}{2}}
\end{array}\right]
\]
for a unitary matrix $U \in \mathbb{O}(d)$ and a vector $b \in \R^{d}$. 
\end{fact}
\Cref{fact:1} can be directly verified by finding the conditions for a real matrix $R$ to be $H$-unitary, i.e., $R^\T H R = H$ or simply $R = H^{-\frac{1}{2}} C H^{\frac{1}{2}} \in \R^{(d+1) \times (d+1)}$ where $C^\T C = I$ and $C \in \C^{(d+1) \times (d+1)}$.
We use this parametric decomposition of rigid transformations to solve the Procrustes problem in $\bL^d$.
\begin{fact} \label{fact:loid_isometry}
$T_{b}^{-1} =T_{-b}$ and $T_{U}^{-1} =T_{U^\T}$ where $b \in \R^d$ and $U \in \mathbb{O}(d)$.
\end{fact}
The hyperbolic translation map $T_b:\bL^d \rightarrow \bL^d$ and hyperbolic rotation map $T_U:\bL^d \rightarrow \bL^d$ are defined as
\begin{align}
T_b(x) &= R_b x,  &&\mbox{for}  \ b \in \R^d, \label{eq:T_b}\\
T_U(x) &= R_U x,   && \mbox{for} \ U \in \mathbb{O}(d). \label{eq:T_U}
\end{align}
\section{Procrustes Analysis} \label{sec:Procrustes_Analysis}
Euclidean (orthogonal) Procrustes analysis proceeds through two steps:
\begin{itemize}
\item Centering: moving the center of mass of both points set to the origin of Cartesian coordinates, and 
\item Finding the optimal rotation/reflection.
\end{itemize}
We proceed to review (and visualize) the definition of the center of mass of a point set in hyperbolic space~\cite[Chapter~13]{mardia2009directional}. 

We start by projecting each point $x \in \bL^d$ onto the following $d$-dimensional subspace
\[
H_d = \set{x \in \R^{d+1}: x_1 = 0}.
\]
Then, we can simply neglect the first element of the projected point (which is always zero), and define a one-to-one map $\mathcal{P}$ between $\bL^d$ and $\R^d$; see \Cref{fig:projection_and_inverse}. In \Cref{def:projection_operator}, we formalize this projection and its inverse. 
\begin{definition} \label{def:projection_operator}
The projection operator $\mathcal{P}: \bL^d \rightarrow \R^d$ and its inverse $\Q$ are defined as
\[
\mathcal{P} \Big( \left[\begin{array}{c}
\sqrt{1+\norm{z}^2} \\ 
z \end{array}\right] \Big) = z, \ \Q (z) = \begin{bmatrix}
\sqrt{1+\norm{z}^2} \\ z
\end{bmatrix}.
\]
\end{definition}
For brevity, we define $\P(X) \bydef [\P(x_1), \ldots, \P(x_N)]$  where $X = [x_1, \ldots, x_N] \in (\bL^d)^N$. Similarly, we consider this extension for $\Q$ as well.

In \Cref{sec:Hyperbolic_Centering}, we review the hyperbolic centering process \cite{mardia2009directional}. In other words, we find a map $T_{b}$ to move the center of mass of projected point sets to $0 \in \R^d$, i.e., $\overline{\P\big( T_{b}(x_n) \big)} = 0$. Then, we show how this centering method helps simplify the hyperbolic Procrustes problem to a sub-problem similar to the famous (Euclidean) orthogonal Procrustes problem.
\begin{figure}[t!]  
	\center
  \includegraphics[width=0.7 \linewidth]{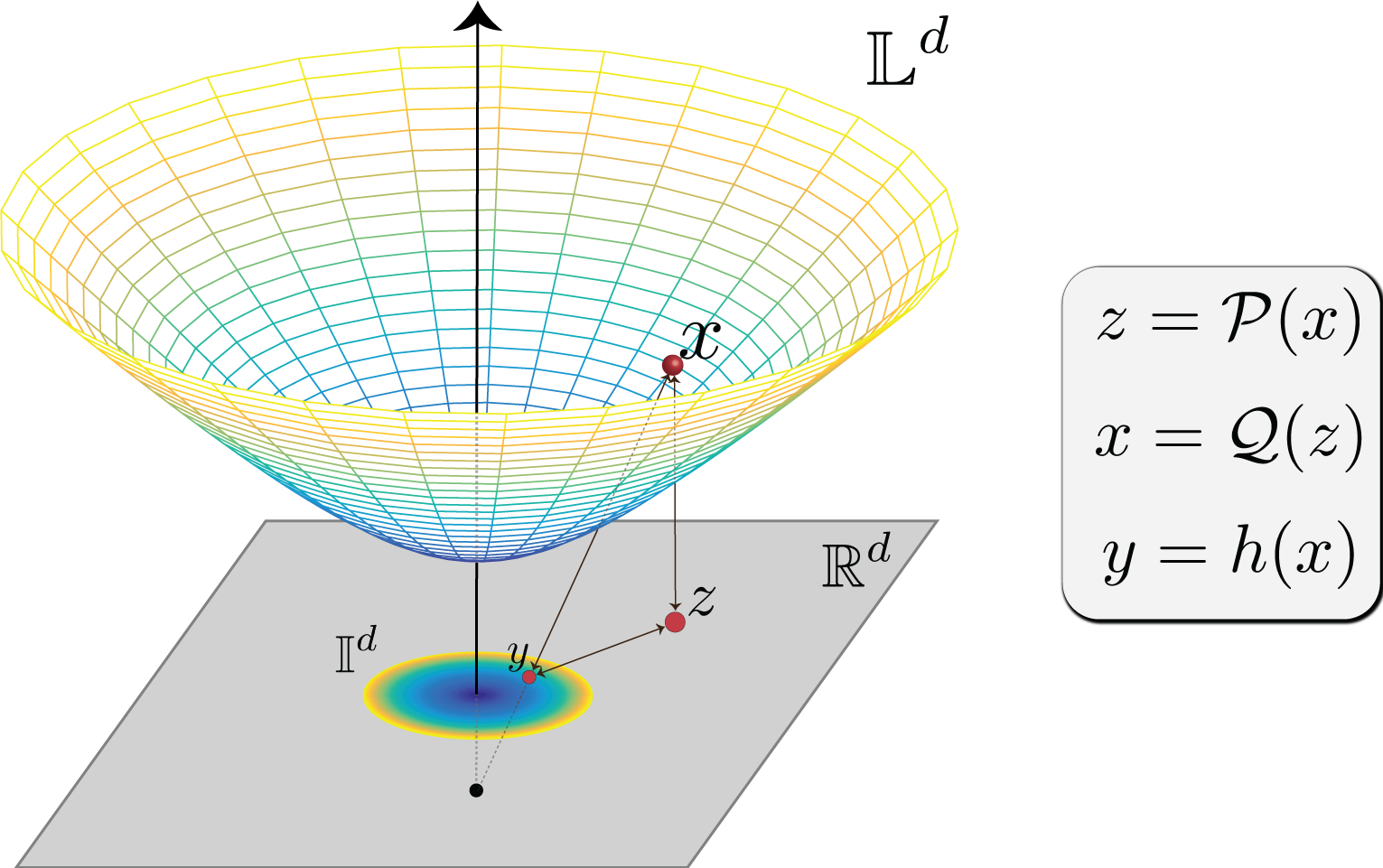}
  \caption{Geometric illustration of $\P$, $\Q$, and stereographic projection $h$.}
  \label{fig:projection_and_inverse}
  \vspace{-10pt}
\end{figure}
\subsection{Hyperbolic Centering}\label{sec:Hyperbolic_Centering}
In Euclidean Procrustes analysis, we have two point sets $z_1, \ldots, z_N$ and $z^{\prime}_1, \ldots z^{\prime}_N$ that are related via a composition of rotation, reflection, and translation maps, i.e.,
\[
z_n = Uz^{\prime}_n + b
\]
where $U \in \mathbb{O}(d)$ and $b \in \R^d$. We extract translation invariant features by moving their point mass to $0 \in \R^d$, i.e.,
\[
z_n - \overline{z_n} = U (z^{\prime}_n -\overline{z^{\prime}_n}).
\]
The main purpose of centering is to map each point set to new locations, $z_n - \overline{z_n}$ and $z^{\prime}_n -\overline{z^{\prime}_n}$ that are invariant with respect to the unknown translation $b$. Subsequently, we can estimate the unknown unitary matrix $\hat{U}$, and then the translation according to $\widehat{b} = \overline{z_n} - \widehat{U} \overline{z^{\prime}_n}$.

In hyperbolic Procrustes analysis, we have
\begin{equation}\label{eq:hyperbolic_procrustes}
x_n = R_b R_U x^{\prime}_n, \forall n \in [N]
\end{equation}
where $U \in \mathbb{O}(d)$ and $b \in \R^d$. In a similar way, we pre-process a point set to extract (hyperbolic) translation invariant locations, i.e., centered point sets. \Cref{lem:centering} gives a simple method to center a projected point set.
\begin{lemma} \cite{mardia2009directional} \label{lem:centering}
Let $x_1, x_2, \ldots, x_N \in \bL^d$. Then, we have
\begin{equation*}
\overline { \P \big(R_{-m_x} x_n \big)}  = 0
\end{equation*}
where $m_x \bydef \frac{1}{\sqrt{-[\overline{x_n}, \overline{x_n}] }} \ \overline{\P (x_n)}$.
\end{lemma}
In~\Cref{prop:hyperbolic_centering}, we show that $T_{-m_x}$ is the canonical translation map for centering the point set $X \in \big(\bL^d\big)^N$.
\begin{proposition}\label{prop:hyperbolic_centering}
Let $x_1, \ldots, x_N$ and $x^{\prime}_1, \ldots x^{\prime}_N$ in $\bL^{d}$ such that
\[
x_n = R_b R_U x^{\prime}_n, \ \forall n \in [N].
\]
for $b \in \R^d$ and $U \in \mathbb{O}(d)$. Then, $R_{-m_{x}}  x_n = R_{V} R_{-m_{x^{\prime}}} x_{n}^{\prime}$ where $R_V$ is a hyperbolic rotation matrix.
\end{proposition}
\begin{proof}
From \Cref{lem:centering}, we have  
\begin{align*}
\overline { R_{-m_x} x_n } = \left[\begin{array}{c}
a_1 \\ 
0 \end{array}\right], \overline { R_{-m_{x^{\prime}}} x^{\prime}_n } = \left[\begin{array}{c}
a_2 \\ 
0 \end{array}\right]
\end{align*}
for $a_1, a_2 \in \R$. On the other hand, we can rewrite \cref{eq:hyperbolic_procrustes} in the following form
\[
R_{-m_{x}}  x_n = R^{\prime}R_{-m_{x^{\prime}}} x_{n}^{\prime}, \forall n \in [N].
\]
where $R^{\prime} =  R_{-m_{x}} R_b R_U R_{m_{x^{\prime}}}$. Since $R^{\prime}$ is an $H$-unitary matrix, we can decompose it as $R^{\prime} = R_{c} R_{V}$ for some $c \in \R^d$ and $V \in \mathbb{O}(d)$. Therefore, we have
\[
\left[\begin{array}{c}
a_1 \\ 
0 \end{array}\right] = R_{c} R_{V} \left[\begin{array}{c}
a_2 \\ 
0 \end{array}\right].
\]
This gives $c = 0$.
\end{proof}
\begin{figure}[b!]
	\center
  \includegraphics[width=1 \linewidth]{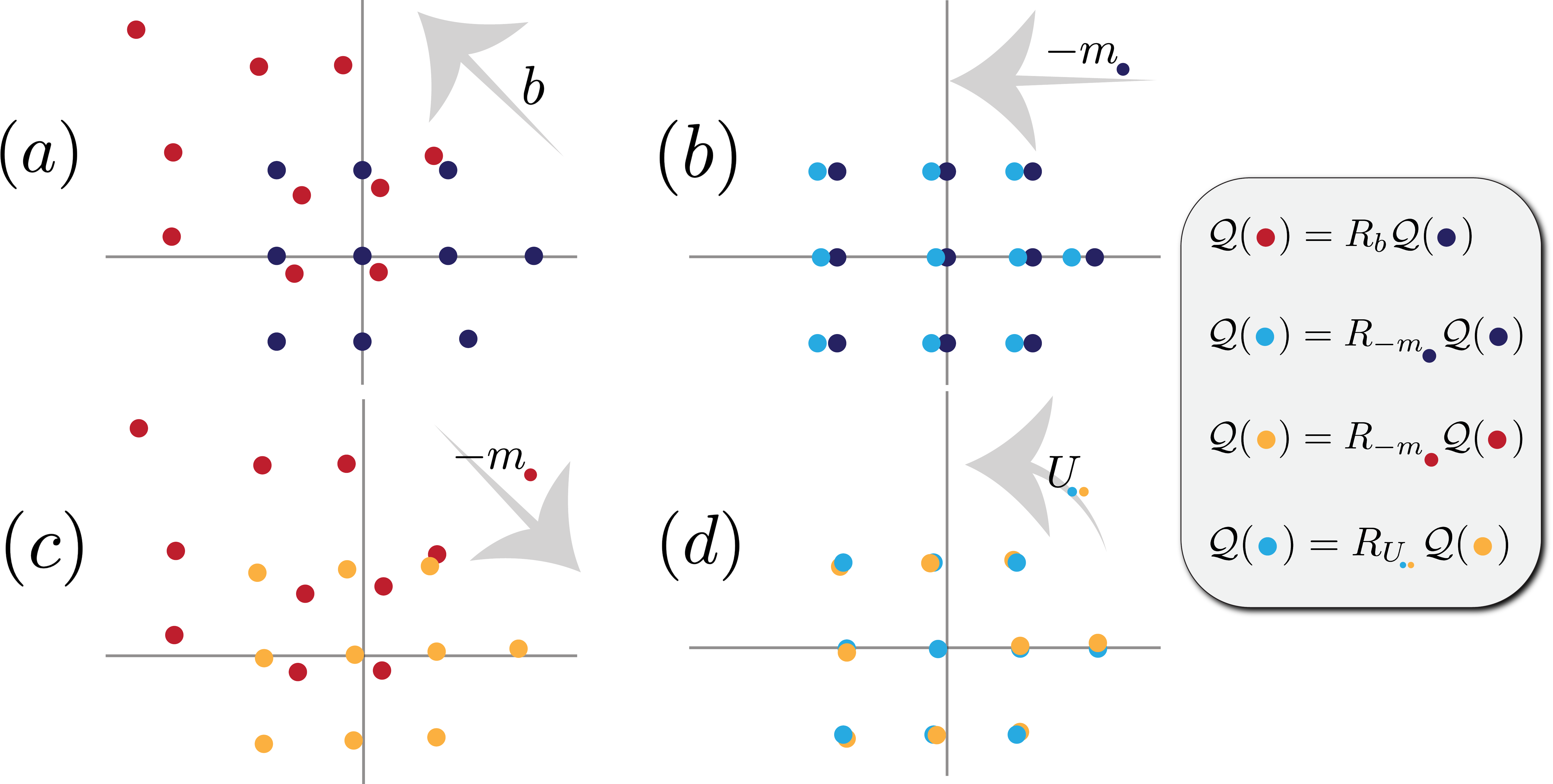}
	  \vspace{-20pt}
  \caption{$(a)$: Red and blue are projected points related by a translation, i.e., $ X= R_b X^{\prime} $. $(b,c)$: Centering each point set. $(d)$: Centered points are related via a rotation, i.e., $R_{m_{x}} R_{b} R_{-m_{x^\prime}} \neq I_d$. }
  \label{fig:translation}
\end{figure}
The map $T_{-m_x}$ not only centers a set of points, but also rotates them. This phenomena is rooted in the noncommutative property of hyperbolic translation or \emph{gyration}. More clearly, for any two vectors $b_1, b_2 \in \R^d$, we have 
\[
R_{b_1} R_{b_2} = R_{V} R_{b_2}R_{b_1}
\]
for a specific unitary matrix $V \in \mathbb{O}(d)$ that accounts for the \emph{gyration} factor; see the example in \Cref{fig:translation} and the follow-up discussion in~\Cref{sec:poincare}. This does not interfere with our analysis since any such rotation is absorbed in $U$, and as we can estimate their \emph{joint} unitary transformation.

Now, let us consider the following noisy case,
\[
x_n = R_b R_U R_{\epsilon_n} x^{\prime}_n, \ \forall n \in [N]
\]
where $\epsilon_n \in \R^d$ is a translation noise for the point $x^{\prime}_n$. Let $z_n = R_{\epsilon_n} x^{\prime}_n$. Then we have $R_{-m_{x}} x_n =  R_{V} R_{-m_{z}} z_n$. The centroid $m_z$ is related to $m_{x^{\prime}} $ and $\set{\epsilon_n}_{n \in [N]}$. Therefore, we can write $m_{z} = m_{x^{\prime}} + \epsilon$ for a $\epsilon \in \R^d$. This leads to
\[
R_{-m_{x}} x_n =  R_{V} R_{\epsilon^{\prime}_n} R_{-m_{x^{\prime}}}x^{\prime}_n, \ \forall n \in [N],
\]
where $R_{\epsilon^{\prime}_n} = R_{-m_{x^{\prime}}-\epsilon} R_{\epsilon_n}R_{m_{x^{\prime}}}$. If the translation noise of each point is sufficiently small, then $R_{V} R_{\epsilon^{\prime}_n} \approx R_{V^{\prime}}$ for a $V^{\prime} \in \mathbb{O}(d)$.
\subsection{Hyperbolic Rotation \& Reflection}
To estimate the unknown hyperbolic rotation, we consider minimizing a weighted discrepancy between the centered point sets. More precisely, 
\begin{equation}\label{eq:cost_function}
\widehat{U} = \argmin_{V \in \mathbb{O}(d)} \sum_{n \in [N]} w_n f \Big( d \big( R_{-m_x} x_n, R_{V}  R_{-m_{x^\prime}}x^{\prime}_n \big) \Big)
\end{equation}
where $d(x,x^{\prime}) = \acosh( -x^\T H x^{\prime})$, $\set{w_n}_{n \in [N]}$ are positive weights, and $f( \cdot ) = \cosh( \cdot )$ is a monotonic function. 
\begin{proposition}
The optimal unitary matrix that solves~\eqref{eq:cost_function} equals $\widehat{U}=U_{l} U_{r}^\T$, where $U_{l} \Sigma U_{r}^\T$ is the singular value decomposition of $\P( R_{-m_x} X) W  \P( R_{-m_{x^\prime}} X^{\prime})^\T$, and $W = \mathrm{diag}(w_1,\ldots, w_N)$.
\end{proposition}
\begin{proof}
We can simplify \eqref{eq:cost_function} as follows:
\[
\widehat{U} = \argmax_{V \in \mathbb{O}(d)} \sum_{n \in [N]} \trace R_{-m_{x^\prime}}x^{\prime}_n  w_n  (R_{-m_x} x_n)^\T H R_{V}.
\]
From \Cref{fact:1}, we know that $R_V$ is only parameterized on its lower right block. The proof then follows from representing the sum in matrix form and invoking von Neumann's trace inequality~\cite{mirsky1975trace}.
\end{proof}

\vspace{-15pt}
\section{M\"{o}bius addition}\label{sec:poincare}
In the Poincar\'e model ($\bI^d$), the points reside in the unit $d$-dimensional Euclidean ball. The isometry between the 'Loid and the Poincar\'e model  $h : \bL^{d} \rightarrow \bI^{d}$ is called the \textit{stereographic projection} \cite{cannon1997hyperbolic}. 
The distance between $y, y^{\prime} \in \bI^d$ is given by $d(y,y^{\prime}) = 2 \mathrm{tanh}^{-1} (\norm{-y \oplus y^{\prime}})$ where $\oplus$ is M\"{o}bius addition ---  a noncommutative and nonassociative operator.
\emph{Gyration} measures the \say{deviation} of M\"{o}bius addition from commutativity,
 i.e., $\mathrm{gyr}[y,y^{\prime}](y^{\prime} \oplus y)  = y \oplus y^{\prime}$ \cite{ungar2008gyrovector}. 
 \begin{fact}
The maps $h \circ R_U \circ h^{-1}$ and $h \circ T_U \circ h^{-1}$ are isometries in the Poincar\'e model, and they can be written as
\begin{align*}
h \circ T_U \circ h^{-1}(y) = Uy, \ \ h \circ T_b \circ h^{-1}(y)  = b^{\prime} \oplus y
\end{align*}
where $b^{\prime} = h \circ \Q(b)$, $T_b$ and $T_U$ are defined in \eqref{eq:T_b} and \eqref{eq:T_U}.
\end{fact}
The translation isometry is a direct result of the Gyrotranslation theorem equality, 
\[
-(c \oplus y) \oplus c \oplus y^{'} =  \mathrm{gyr}[c, y](-y \oplus y^{\prime}),
\]
where $c \in \bI^d$ \cite{ungar2008gyrovector}. Therefore, left M\"{o}bius addition preserves the distances of point sets in the Poincar\'e model\footnote{M\"{o}bius gyrations hence keep the norm that they inherit from $\R^d$ invariant, i.e., $\norm{\mathrm{gyr}[c, y] (-y \oplus y^{\prime})} = \norm{-y \oplus y^{\prime}}$ \cite{ungar2008gyrovector}.}. 
We can hence perform a Procrustes analysis in the Poincar\'e model by $(1)$ centering each point set, i.e., subtracting their center of mass from the left hand side of the M\"{o}bius addition, and $(2)$ estimating the remaining rotation factor --- a composition of gyrations and the initial unknown rotation between the two point sets.
\enlargethispage{\baselineskip}
\enlargethispage{\baselineskip}
\section{Numerical Analysis}
Let $x_n = R^* R_{\epsilon_n} x^{\prime}_n, \ \forall n \in [N]$ where $R^*$ is an $H$-unitary matrix and $\epsilon_1, \ldots, \epsilon_N$ is the set of translation noise samples. 

We compute the following $H$-unitary operators to match the point sets $X, X^{\prime}$:
\begin{itemize}
\item $R_P$: The matrix estimated by our proposed method;
\item $R_{\mathrm{GD}}$: Let $e(X, \overline{X}) \bydef \frac{1}{Nd} \sum_{n \in [N]}d(x_n, \overline{x}_n)$ be the normalized discrepancy between $X$ and $\overline{X}$. The matrix $R_{\mathrm{GD}}$ is computed by an iterative gradient descent method: We initialize $R_{\mathrm{GD}} = I_{d+1}$, and iterate the following steps: $(1)$ $\widehat{b} = - \alpha \frac{\partial}{\partial b} e(X,R_b R_{\mathrm{GD}} X^{\prime})|_{b=0}$ for a small $\alpha >0$; $(2)$ $\widehat{U} = \argmax_{U \in \mathbb{O}(d)} \sum_{n \in [N]} [x_n, R_U R_{\widehat{b}} R_{\mathrm{GD}} x^{\prime}_n]$; $(3)$ Update $R_{\mathrm{GD}} \leftarrow R_{\widehat{U}} R_{\widehat{b}} R_{\mathrm{GD}}$;
\item $R_{\mathrm{GD+P}}$: We can combine the aforementioned methods by $(1)$ solving the problem with our method, and $(2)$ fine-tuning the estimated isometry by applying the gradient method on the point sets $X$ and $R_P X^{\prime}$. 
\end{itemize}
For a random $H$-unitary $R^*$ and all $n \in [N]$, we sample $d$-dimensional $z_n \sim \mathcal{N}(0,I)$, and $\epsilon_n \sim 10^{-2} \mathcal{N}(0,I)$; Then, we let $x^{\prime}_n = \mathcal{Q}(z_n)$ and $x_n = R^*R_{\epsilon_n} x^{\prime}_n$. For $10^3$ random $(X, X^{\prime})$ pairs, we compute their normalized discrepancy $e(X, R X^{\prime})$, where $R \in \set{R_{P}, R_{\mathrm{GD}},R_{\mathrm{GD+P}}}$.\footnote{For our method, we choose $W = 11^\T$.} All methods successfully denoise the measurements, i.e., $e(X, R^* X^{\prime}) > e(X, R X^{\prime})$; see \Cref{fig:analysis} $(a)$. We should note that the gradient descent method does not necessarily converge to an acceptable solution. Therefore, we report the number of outlier trials, i.e.,
\begin{equation}\label{eq:outlier}
(X,X^{\prime}) : |e(X, R X^{\prime}) - Q_2 | > k \frac{1}{2}|Q_3-Q_1|
\end{equation}
where $Q_1,Q_2$ and $Q_3$ are first, second and third quartiles of the total reported discrepancies, and $k=5$ for a conservative criterion to pick outliers (see ~\Cref{fig:analysis} $(b)$). The gradient descent method has the most number of outlier whereas our proposed method has the minimum number of outliers --- comparable to outliers in the measurement noise.  Therefore, the proposed method robustly solves the hyperbolic Procrustes problem and its accuracy can be moderately improved  with a post fine-tuning gradient method.
\begin{figure}[t] 
	\center
  \includegraphics[width=1 \linewidth]{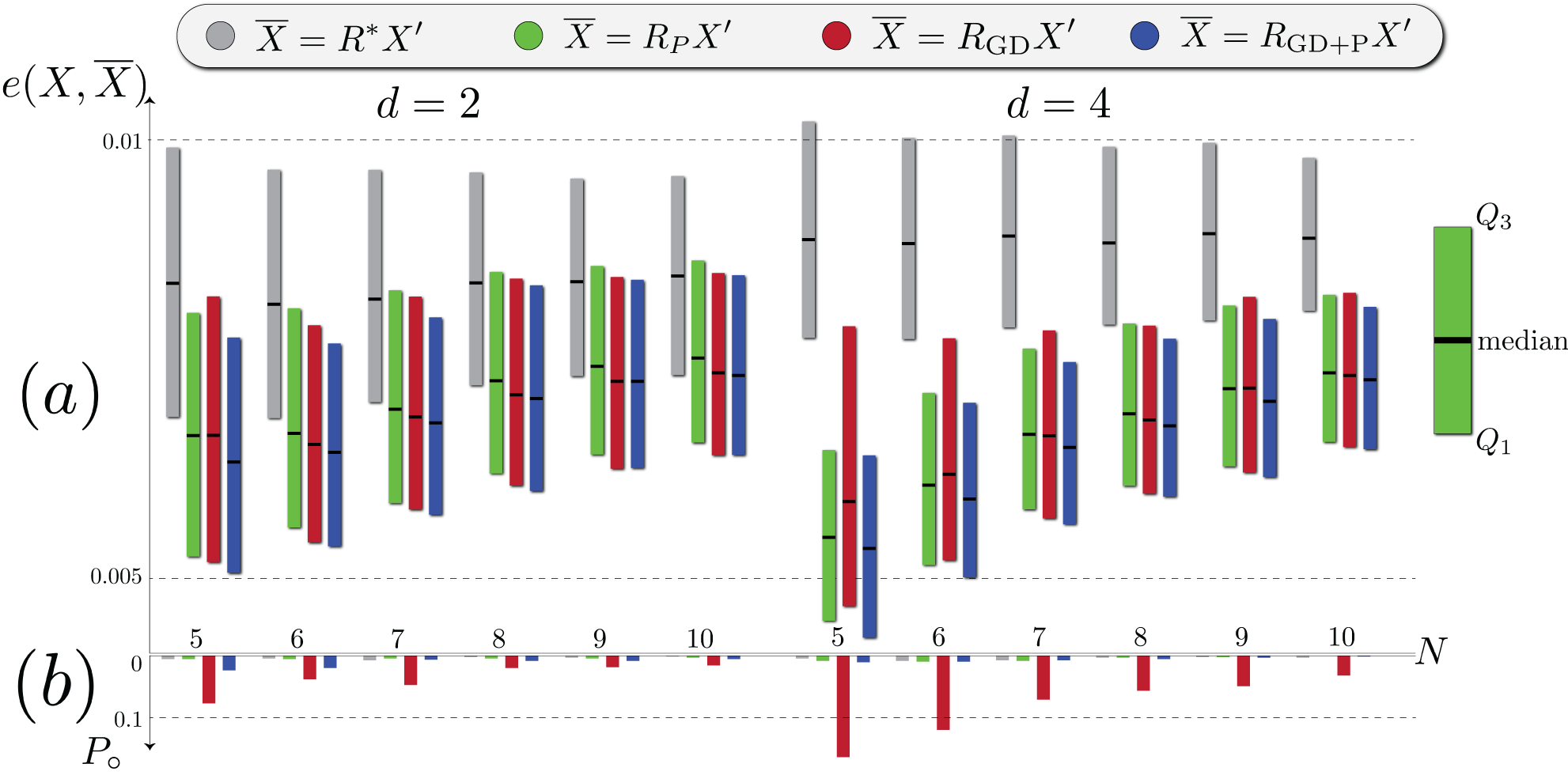}
	 \vspace{-15pt}
  \caption{$(a)$ Normalized discrepancy for random hyperbolic point sets of size $N \in \set{5, \ldots, 10}$ and dimensions $d \in \set{2,4}$. For $10^3$ trials, we report the quartiles $Q_1,Q_2$ and $Q_3$ since they are robust to outliers. $(b)$ The probability of an outlier event $P_\circ = 10^{-3} \times \mbox{total number of outliers}$, e.g., the fraction of examples that failed to converge or outlier defined in the sense of ~\eqref{eq:outlier}.}
  \label{fig:analysis}
  	 \vspace{-10pt}
\end{figure}
\vspace{-5pt}
\section{Conclusion}
Inspired by its Euclidean counterpart, we introduced the Procrustes problem in hyperbolic spaces. We reviewed the (indefinite) Lorentizian inner product, and described how $H$-unitary matrices represent isometries in the 'Loid model of hyperbolic spaces. Using the  parameterized decomposition of hyperbolic isometries in terms of hyperbolic rotation and translation, we showed that moving the \emph{center of mass} to the origin gives point sets that are invariant to hyperbolic translation (for the case of no measurement noise). We then used the centered point sets to estimate the unknown rotation factor. 
\vspace{-15pt}
\section{Acknowledgment}
The authors would like to thank Prof. Olgica Milenkovic for helpful discussions and suggestions.
\bibliography{bare_jrnl}

\end{document}